\documentclass[11pt]{amsart}

\usepackage[margin=1in]{geometry}
\usepackage{amsmath, amssymb, amsthm}
\usepackage{graphicx}
\usepackage{hyperref}
\hypersetup{hidelinks,
  pdftitle={Generalized Fidelity, the Data Processing Inequality, and Convexity},
  pdfauthor={Haobo Ma, Reza Rajaei, Wenlin Zhang}}
\usepackage{setspace}
\usepackage{color}

\DeclareMathOperator{\Tr}{Tr}

\title{Generalized Fidelity, the Data Processing Inequality, and Convexity}
 
\author{Haobo Ma}
\address{ChronoAI Pte Ltd, Singapore; The Omega Institute}
\email{auric@aelf.io}

\author{Reza Rajaei}
\address{David and Judi Proctor Department of Mathematics, University of Oklahoma, Norman, OK}
\email{reza.rajaei-1@ou.edu}

\author{Wenlin Zhang}
\address{National University of Singapore, Singapore; The Omega Institute}
\email{e1327962@u.nus.edu}

\date{}

\newtheorem{theorem}{Theorem}
\newtheorem{lemma}{Lemma}
\newtheorem{corollary}{Corollary}

\begin{document}

\begin{abstract} In this paper, we show that the generalized Bures--Wasserstein distance induced by the generalized fidelity fails the data processing inequality (DPI) in every dimension $n\geq2$. In dimension two, we derive three sufficient conditions guaranteeing the DPI by using an explicit formula for the real part of the generalized fidelity. Furthermore, we find a necessary and sufficient condition under which the generalized fidelity reduces to the Uhlmann fidelity. We also show that neither the generalized Bures--Wasserstein distance nor its square is convex in one state, even with the other state and the reference fixed, in every dimension $n\geq2$. 
 
\end{abstract}

\maketitle

\section*{Introduction}
In the context of quantum information theory, various types of fidelity
functions have been introduced to quantify the similarity of quantum states \cite{Liang2019}. Besides the Uhlmann fidelity
\cite{Uhlmann1976,Jozsa1994}, examples include the Holevo affinity
\cite{Holevo1972}, Matsumoto's reverse-test fidelity \cite{Matsumoto2010},
and alternatives based on Hilbert--Schmidt overlaps and purities
\cite{Wang2008,Mendonca2008}.

We use the root (unsquared) Uhlmann fidelity
$F^U(P,Q)=\|\sqrt P\sqrt Q\|_1$ throughout. Uhlmann's transition
probability and Jozsa's fidelity convention are its square
\cite{Uhlmann1976,Jozsa1994}. The root fidelity is the maximal overlap of
purifications \cite{Uhlmann1976,Jozsa1994}, and it is nondecreasing under
completely positive trace-preserving (CPTP) maps
\cite[Theorem~3.27]{Watrous}. Its relation to state distinguishability is
quantified by the Fuchs--van de Graaf inequalities \cite{FuchsVanDeGraaf1999}: for density matrices and
$T(P,Q)=\tfrac12\|P-Q\|_1$,
\[
 1-F^U(P,Q)\leq T(P,Q)\leq\sqrt{1-F^U(P,Q)^2}.
\]
These properties make fidelity a standard tool
in quantum information theory \cite{NielsenChuang2010}.

The corresponding Bures distance \cite{Bures1969} has the matrix form
\[
 d_{\mathrm{BW}}(P,Q)^2=\Tr(P+Q)-2F^U(P,Q).
\]
Bhatia, Jain and Lim develop its geometry on positive definite matrices
\cite{BhatiaJainLim2019}; for real covariance matrices, the same expression
is the squared Wasserstein distance of order two between centered Gaussian
distributions \cite{DowsonLandau1982}. Its quantum-state geometry is discussed
in \cite{BengtssonZyczkowski2017}. Related matrix Hellinger distances have
also been studied \cite{BhatiaGaubertJain2019}; the generalized Hellinger
divergences of Pitrik and Virosztek are jointly convex and satisfy data
processing \cite{PitrikVirosztek2020}. Quantum transport also includes
couplings defined through quantum channels \cite{DePalmaTrevisan2021} and
a Wasserstein distance of order one reflecting tensor-product locality
\cite{DePalmaEtAl2021}. These constructions provide distinct ways to extend
classical distances to quantum states.

Data processing expresses the loss of distinguishability under a channel.
For relative entropy, it was established by Lindblad \cite{Lindblad1975};
Petz's quasi-entropies \cite{Petz1986} and classification of monotone
Riemannian metrics \cite{Petz1996} place such compatibility with channels in
a wider framework. Parameterized constructions include the
$\alpha$-$z$ R\'enyi relative entropies of Audenaert and Datta
\cite{AudenaertDatta2015} and the generalized $\alpha$-$z$ fidelities studied
by Yan et al.\ \cite{Yan2025}. For the sandwiched R\'enyi divergence,
Beigi proves data processing for $\alpha>1$ \cite{Beigi2013}, while
Frank and Lieb cover $\alpha\geq\tfrac12$ (with the usual limit at
$\alpha=1$) \cite{FrankLieb2013}. At $\alpha=\tfrac12$ this divergence is
$\widetilde D_{1/2}(P\|Q)=-2\log F^U(P,Q)$, so its data processing
recovers the fidelity inequality. A metric property alone does not imply
contraction under channels.

Afham and Ferrie introduced a generalized fidelity with a reference
parameter. Suitable references recover the Uhlmann, Holevo and Matsumoto
fidelities \cite{MR4942599}. For positive semidefinite matrices
$P,Q$ and a positive definite reference $R$ of the same size, it is
\begin{equation}
F_R(P,Q)=\Tr\!\left[\sqrt{R^{1/2}PR^{1/2}}R^{-1}\sqrt{R^{1/2}QR^{1/2}}\right].
\end{equation}
The reference must be positive definite, whereas $P$ and $Q$ may be
rank-deficient. The generalized fidelity can be complex. When $P=R$, it
reduces to the root Uhlmann fidelity,
\begin{equation}
F_R(R,Q)=F^U(R,Q)=\Tr\!\left(\sqrt{R^{1/2}QR^{1/2}}\right).
\end{equation}
Afham and Ferrie proved that $b_R(P,Q)=\sqrt{B_R(P,Q)}$ is a metric on
the positive definite cone, where
$B_R(P,Q)=\Tr(P+Q)-2\Re F_R(P,Q)$ \cite[Theorem~3.5]{MR4942599}.
The extension to positive semidefinite $P,Q$ used here follows directly from
\begin{equation}\label{eq:metric-embedding}
 f_R(P)=R^{-1/2}\sqrt{R^{1/2}PR^{1/2}},\qquad
 B_R(P,Q)=\|f_R(P)-f_R(Q)\|_{\mathrm{HS}}^2.
\end{equation}
Indeed, $\|f_R(P)\|_{\mathrm{HS}}^2=\Tr P$ and the cross term is
$\Re F_R(P,Q)$. The map $f_R$ is injective: equality of its values gives
equality of the positive square roots, hence of $P,Q$. Thus its Hilbert--Schmidt
distance is a metric on the whole positive semidefinite cone.

A recent preprint by Vuong, \emph{Polar Fidelities, Holevo Bases, and Unitary
Factors of Generalized Fidelity}, addresses polar-path realizability, the
fixed-pair Holevo-base equation, and the unitary factors attainable from a
generalized-fidelity base \cite{Vuong2026}. Proposition~4.6 there also proves
partial data processing for symmetrized polar fidelities with parameter
$x\in[-1,1]$, for positive definite input and output pairs, when the input
or output pair commutes, or when the channel factors through a commuting
range with positive definite intermediate states. Those polar fidelities
choose their references from the state pair; in the present DPI, an
independently specified reference $R$ is transported to $\Phi(R)$.
In general $\Phi(P^x)\neq\Phi(P)^x$, so the two reference prescriptions
give different inequalities.
We disprove the unrestricted data processing inequality considered in the
first open problem of \cite[Sec.~X]{MR4942599}. In dimension two, we express the real
part of the generalized fidelity through pairwise Uhlmann fidelities and give
three sufficient conditions for data processing; Corollary~\ref{cor:qubit-dpi-band}
subsumes the other two. We also disprove unrestricted convexity in the state variables,
as considered in the second open problem of \cite[Sec.~X]{MR4942599}: neither $B_R(P,Q)$ nor $b_R(P,Q)$ is convex
in $P$ in any dimension $n\geq2$, even with $Q$ and $R$ fixed, so neither is
jointly convex.

This paper revises and extends Rajaei's earlier single-author preprint
\cite{Rajaei2026v1}, which contains the qubit fidelity formula, the equality
characterization, the higher-dimensional DPI construction, and the first
two sufficient qubit DPI conditions. Rajaei also constructed the
higher-dimensional convexity example in the present paper. Ma and Zhang
contributed the explicit qubit counterexamples in
Theorems~\ref{thm:qubit-dpi-counterexample} and
\ref{thm:qubit-convexity-counterexample}, and the stronger sufficient criterion
in Corollary~\ref{cor:qubit-dpi-band}, separately from the higher-dimensional
constructions. The measure-and-prepare DPI construction in
Theorem~\ref{thm:higher-dpi} uses three distinguishable input directions and
therefore begins in dimension three. This restriction concerns the channel
construction: its leading $2\times2$ blocks already have negative trace
$\Tr(A_2C_2B_2)=-6$.

\section*{The Qubit Case}

In this section, we present another formulation of the definition above in the qubit case, which is free of matrix square roots. This new formulation enables us to have a simpler understanding of the generalized fidelity. We begin with an elementary consequence of Cayley--Hamilton.  

\noindent\begin{minipage}{\linewidth}
\begin{lemma}[Cayley--Hamilton formula for square roots]
    The square root of a nonzero $2 \times2$ PSD matrix A is:

    $$\sqrt A= \frac{A+ \sqrt{\det(A)}I_{2}}{\sqrt{\operatorname{Tr}(A)+2 \sqrt{\det(A)} }}.$$

\end{lemma}
\end{minipage}\par
\begin{proof}
Cayley--Hamilton gives $A^2-(\Tr A)A+(\det A)I_2=0$, hence
\[
 (A+\sqrt{\det A}I_2)^2=(\Tr A+2\sqrt{\det A})A.
\]
The numerator is positive semidefinite and the denominator is positive for
nonzero $A\succeq0$, proving the formula.
\end{proof}

For nonzero $P,Q\succeq0$ and $R\succ0$ in dimension two, this lemma gives

\[
\begin{aligned}
F_R(P,Q)
&= \operatorname{Tr}\left[
    \frac{
        R^{1/2}PR^{1/2}+\sqrt{\det(R^{1/2}PR^{1/2})}I_{2}
    }{
        \sqrt{\operatorname{Tr}(R^{1/2}PR^{1/2})
        +2\sqrt{\det(R^{1/2}PR^{1/2})}}
    }
\right. \\[6pt]
&\qquad\left.
    R^{-1}
    \frac{
        R^{1/2}QR^{1/2}+\sqrt{\det(R^{1/2}QR^{1/2})}I_{2}
    }{
        \sqrt{\operatorname{Tr}(R^{1/2}QR^{1/2})
        +2 \sqrt{\det(R^{1/2}QR^{1/2})}}
    }
\right] \\[6pt]
&= \frac{1}{
    \sqrt{\operatorname{Tr}(RP)+2\sqrt{\det(RP)}}
    }
    \frac{1}{
    \sqrt{\operatorname{Tr}(RQ)+2\sqrt{\det(RQ)}}
    } \\[6pt]
&\qquad
\operatorname{Tr}\left[
    \left(R^{1/2}PR^{1/2}
    +\sqrt{\det(R^{1/2}PR^{1/2})}I_{2}\right)R^{-1}
\right. \\[4pt]
&\qquad\quad\left.
    \left(R^{1/2}QR^{1/2}
    +\sqrt{\det(R^{1/2}QR^{1/2})}I_{2}\right)
\right],
\end{aligned}
\]

which simplifies to:

\[
\begin{aligned}
&\frac{1}{
    \sqrt{\operatorname{Tr}(RP)+2\sqrt{\det(RP)}}
}
\frac{1}{
    \sqrt{\operatorname{Tr}(RQ)+2\sqrt{\det(RQ)}}
}
\\[4pt]
&\qquad
\Big(
    \operatorname{Tr}(RPQ)
    + \operatorname{Tr}(Q)\sqrt{\det(RP)}
\\
&\qquad\qquad
    + \operatorname{Tr}(P)\sqrt{\det(RQ)}
    + \operatorname{Tr}(R^{-1})\sqrt{\det(R^2PQ)}
\Big).
\end{aligned}
\]

Using the fact that for an invertible $2 \times 2$ matrix $R$, we have $\Tr(R^{-1})= \frac{\Tr(R)}{\det (R)}$, we obtain:

\begin{equation}\label{eq:qubit-quotient}
\begin{aligned}
&F_R(P,Q)=\frac{1}{
    \sqrt{\operatorname{Tr}(RP)+2\sqrt{\det(RP)}}
}
\frac{1}{
    \sqrt{\operatorname{Tr}(RQ)+2\sqrt{\det(RQ)}}
}
\\[4pt]
&\qquad
\Big(
    \operatorname{Tr}(RPQ)
    + \operatorname{Tr}(Q)\sqrt{\det(RP)}
\\
&\qquad\qquad
    + \operatorname{Tr}(P)\sqrt{\det(RQ)}
    + \operatorname{Tr}(R)\sqrt{\det(PQ)}
\Big).
\end{aligned}
\end{equation}

If $P=0$ or $Q=0$, the defining formula instead gives
$F_R(0,Q)=F_R(P,0)=0$; the quotient in
\eqref{eq:qubit-quotient} is not used in these cases.

\begin{corollary} In the qubit case, $F_R(P,Q)$ is zero if and only if 

\[
\begin{aligned}
&\qquad
    \operatorname{Tr}(RPQ)
    + \operatorname{Tr}(Q)\sqrt{\det(RP)}
    + \operatorname{Tr}(P)\sqrt{\det(RQ)}
    + \operatorname{Tr}(R)\sqrt{\det(PQ)}
\end{aligned}
\]
 is zero. If either operator is zero, both this numerator and the defining fidelity vanish directly. In particular, for pure qubit states, $F_R(P,Q)$ is zero if and only if $\operatorname{Tr}(RPQ)$ is zero. 
\end{corollary}
    
The second part extends to arbitrary dimensions by the pure-state formula of Afham and Ferrie \cite[Property~4]{MR4942599}:

\begin{corollary} For two pure states $P$ and $Q$, $F_R(P,Q)$ is zero if and only if $\operatorname{Tr}(RPQ)$ is zero.

\end{corollary}

 \begin{proof}
     Since $P$ and $Q$ are pure states, both $R^{\frac{1}{2}}PR^{\frac{1}{2}}$ and $R^{\frac{1}{2}}QR^{\frac{1}{2}}$ are rank one. So $\sqrt{R^{\frac{1}{2}}PR^{\frac{1}{2}}}=\frac{R^{\frac{1}{2}}PR^{\frac{1}{2}}}{\sqrt{\Tr(R^{\frac{1}{2}}PR^{\frac{1}{2}})}}$ and $\sqrt{R^{\frac{1}{2}}QR^{\frac{1}{2}}}=\frac{R^{\frac{1}{2}}QR^{\frac{1}{2}}}{\sqrt{\Tr(R^{\frac{1}{2}}QR^{\frac{1}{2}})}}.$ Therefore, 

     \[
F_R(P,Q)
=
\operatorname{Tr}\!\left[
\sqrt{R^{1/2} P R^{1/2}}\,
R^{-1}\,
\sqrt{R^{1/2} Q R^{1/2}}
\right]
=
\frac{\operatorname{Tr}(RPQ)}
{\sqrt{\operatorname{Tr}(RP)}\,\sqrt{\operatorname{Tr}(RQ)}}.
\]

\end{proof}

At this point, we mention four lemmas that allow us to turn the real part of $F_R(P,Q)$ into an expression in terms of the Uhlmann fidelities between $R, P,$ and $Q.$  

\begin{lemma}

\label{lem:trace of three}

Assume qubit density matrices $R, P,$ and $Q$ are written in terms of the Pauli matrices, namely:
$$R=\frac{1}{2}(I_2+\vec{r} \cdot \vec{\sigma}),$$
$$P=\frac{1}{2}(I_2+\vec{p} \cdot \vec{\sigma}),$$
$$Q=\frac{1}{2}(I_2+\vec{q} \cdot \vec{\sigma}).$$

Then, 

$$\Tr(RPQ)= \frac{1}{4}(1+\vec{r}\cdot\vec{p}+\vec{q}\cdot\vec{r}+\vec{p}\cdot\vec{q} + i\vec{q}\cdot(\vec{r} \times \vec{p})),$$

where $\cdot$ and $\times$ respectively denote the inner and cross products. 

\begin{proof}
    Using $(\vec{r} \cdot \vec{\sigma})(\vec{p} \cdot \vec{\sigma})=(\vec{r} \cdot \vec{p})I_2+ i(\vec{r} \times \vec{p}) \cdot \vec{\sigma}$, we conclude that:

    $$RP=\frac{1}{4}(I_2+\vec{r} \cdot \vec{\sigma} +\vec{p} \cdot \vec{\sigma}+(\vec{r} \cdot \vec{p})I_2+i(\vec{r} \times \vec{p}) \cdot \vec{\sigma}),$$

    and consequently, 

\[
\begin{aligned}
RPQ=\frac{1}{8}\Big(&I_2+\vec{r}\cdot\vec{\sigma}
+\vec{p}\cdot\vec{\sigma}
+(\vec{r}\cdot\vec{p})I_2
+i(\vec{r}\times\vec{p})\cdot\vec{\sigma}\\
&+\vec{q} \cdot \vec{\sigma}+(\vec{r}\cdot\vec{q})I_2
+i(\vec{r}\times\vec{q})\cdot\vec{\sigma}
+(\vec{p}\cdot\vec{q})I_2
+i(\vec{p}\times\vec{q})\cdot\vec{\sigma}\\
&+(\vec{r}\cdot\vec{p})\vec{q}\cdot\vec{\sigma}
+i\Big((\vec{q}\cdot(\vec{r}\times\vec{p}))I_2
+i((\vec{r}\times\vec{p})\times\vec{q})\cdot \vec{\sigma}\Big)\Big).
\end{aligned}
\]

Since each component of $\sigma$ has trace zero and the trace of $I_2$ is 2, all the terms including $\vec{\sigma}$ disappear and we obtain that:

$$\Tr(RPQ)= \frac{1}{4}(1+\vec{r}\cdot\vec{p}+\vec{q}\cdot\vec{r}+\vec{p}\cdot\vec{q} + i\vec{q}\cdot(\vec{r} \times \vec{p})).$$
    
\end{proof}
\end{lemma}

Considering the imaginary part of $\Tr(RPQ)$ in Lemma~\ref{lem:trace of three}, we derive a corollary. 

\begin{corollary}
    Under the setting of Lemma~\ref{lem:trace of three}, $F_R(P, Q)$ is a real number if and only if $\vec{p}, \vec{q}, $ and $\vec{r}$ are linearly dependent, equivalently contained in a plane through the origin. 
\end{corollary}
\begin{lemma}
    For qubit states $P$ and $Q$, we have:

    $$F^U(P, Q)= \sqrt{\Tr(PQ)+2\sqrt{\det(PQ)}}.$$

    \begin{proof}
This is the square root of the squared-fidelity formula in
\cite[Sec.~III~A, Eq.~(13)]{Mendonca2008}. Directly, if $\lambda_1,\lambda_2$
are the eigenvalues of $H=\sqrt P Q\sqrt P\succeq0$, then
$(\Tr\sqrt H)^2=\Tr H+2\sqrt{\det H}$, which gives the stated identity.
\end{proof}

\end{lemma}

\begin{lemma}

\label{lem:fidelity}
   Assume qubit density matrices $P$ and $Q$ are written in terms of the Pauli matrices, namely:

$$P=\frac{1}{2}(I_2+\vec{p}\cdot \vec{\sigma}),$$
$$Q=\frac{1}{2}(I_2+\vec{q}\cdot \vec{\sigma}).$$

Then, 

$$(F^U(P,Q))^2=\frac{1}{2}(1+\vec{p}\cdot\vec{q}+\sqrt{(1-||p||^2)(1-||q||^2)}).$$ 

\begin{proof}
Substitute $\Tr(PQ)=(1+\vec p\cdot\vec q)/2$ and
$\det P=(1-\|\vec p\|^2)/4$, and similarly for $Q$, into the preceding
lemma. The determinant identity is proved in Lemma~\ref{lem:det}.
\end{proof}

\end{lemma}

\begin{lemma}

\label{lem:det}
    For a qubit density matrix $P=\frac{1}{2}(I_2+\vec{p}\cdot \vec{\sigma}),$ $\det(P)=\frac{1-||\vec{p}||^2}{4}$.

    \begin{proof}
For a qubit density matrix
\[
P=\frac{1}{2}\left(I_2+\vec{p}\cdot\vec{\sigma}\right),
\]
where $\vec{p}=(p_x,p_y,p_z)$ is the Bloch vector, we have:
\[
P=
\frac{1}{2}
\begin{pmatrix}
1+p_z & p_x-ip_y\\
p_x+ip_y & 1-p_z
\end{pmatrix}.
\]
Therefore,
\[
\begin{aligned}
\det(P)
&=\frac{1}{4}\left[(1+p_z)(1-p_z)
-(p_x-ip_y)(p_x+ip_y)\right]\\
&=\frac{1}{4}\left(1-p_x^2-p_y^2-p_z^2\right)\\
&=\frac{1-\|\vec{p}\|^2}{4}.
\end{aligned}
\]

    \end{proof}
\end{lemma}

Note that Lemma~\ref{lem:fidelity} and Lemma~\ref{lem:det} imply that:

$$(F^U(P,Q))^2=\frac{1}{2}(1+\vec{p}\cdot\vec{q})+2\sqrt{\det(PQ)}.$$

Therefore, for qubit states, we will have:

\[
\begin{aligned}
\operatorname{Re}\!\left(F_R(P,Q)\right)
&=
\operatorname{Re}\!\Bigg[
\frac{1}{
    \sqrt{\operatorname{Tr}(RP)+2\sqrt{\det(RP)}}
}
\frac{1}{
    \sqrt{\operatorname{Tr}(RQ)+2\sqrt{\det(RQ)}}
}
\\[4pt]
&\qquad\qquad
\Big(
    \operatorname{Tr}(RPQ)
    + \operatorname{Tr}(Q)\sqrt{\det(RP)}
    + \operatorname{Tr}(P)\sqrt{\det(RQ)}
    + \operatorname{Tr}(R)\sqrt{\det(PQ)}
\Big)
\Bigg]
\\[6pt]
&=
\frac{
    \frac14
    \left(
        1+\vec r\!\cdot\!\vec p
        +\vec q\!\cdot\!\vec r
        +\vec p\!\cdot\!\vec q
    \right)
    +\sqrt{\det(RP)}
    +\sqrt{\det(RQ)}
    +\sqrt{\det(PQ)}
}{
    F^U(R,P)\,F^U(R,Q)
}, 
\end{aligned}
\]

which a simple application of the lemmas above indicates is equal to:

\[
\frac{
\bigl(F^U(P,Q)\bigr)^2
+\bigl(F^U(R,P)\bigr)^2
+\bigl(F^U(R,Q)\bigr)^2
-1
}{
2F^U(R,P)\,F^U(R,Q)
}.
\]

We summarize the last result as a corollary. 

\begin{corollary}

\label{cor:fidelity}
    For qubit states, it holds that: 

    \[
\operatorname{Re}\!\left(F_R(P,Q)\right)
=
\frac{
\bigl(F^U(P,Q)\bigr)^2
+\bigl(F^U(R,P)\bigr)^2
+\bigl(F^U(R,Q)\bigr)^2
-1
}{
2F^U(R,P)\,F^U(R,Q)
}.
\]

\end{corollary}

\section*{Reduction to the Uhlmann Fidelity}

In this section, we find a necessary and sufficient condition under which the generalized fidelity reduces to the Uhlmann fidelity. For this purpose, we need two lemmas.

\begin{lemma}\label{lem:trace-norm-equality}
  Assume $M$ is a square matrix. Then $\left\|M\right\|_{1}
 =|\Tr(M)|$ if and only if $M=e^{i \theta} A$ for some $\theta \in \mathbb{R}$, where $A$ is a PSD matrix.

\begin{proof}
    The case $M=0$ is immediate. Otherwise, using the SVD decomposition, $M$ can be written as:

$$M= \sum_{i=1}^{r} \sigma_iu_iv_i^*,$$

where $u_1, ..., u_r$ and $v_1, ..., v_r$ are two orthonormal sets and $\sigma_1, ..., \sigma_r $ are the non-zero singular-values. Then, 

\[
\begin{aligned}
|\operatorname{Tr}(M)|
&= \left|\operatorname{Tr}\left(\sum_{i=1}^r \sigma_i u_i v_i^*\right)\right|
\leq \sum_{i=1}^r \sigma_i |\operatorname{Tr}(u_i v_i^*)|
= \sum_{i=1}^r \sigma_i |v_i^*u_i| \\
&\leq \sum_{i=1}^r \sigma_i
= \left\|M\right\|_{1}
.
\end{aligned}
\]

Hence, the equality case occurs if and only if $\operatorname{Tr}(u_iv_i^*)=e^{i\theta}$ for some $\theta \in \mathbb{R}$, for each $i$, which implies $$\operatorname{Tr}(u_iv_i^*)=v_i^*u_i=e^{i\theta}.$$ So, $$u_i=e^{i\theta}v_i.$$ Therefore, $$M= \sum_{i=1}^{r} \sigma_iu_iv_i^*=\sum_{i=1}^{r} \sigma_ie^{i\theta}v_iv_i^*,$$ which yields the claim. 
\end{proof}
\end{lemma}
\begin{lemma}

\label{lem:uhlmann-factorization}
Let $A, B$ be square matrices of the same size, then it holds that:
$$F^U(AA^*, BB^*)=\left\|A^{*}B\right\|_{1}.$$

\end{lemma}

\begin{proof}
    See Lemma $3.21$ in \cite{Watrous}. 
\end{proof}

The bound $|F_R(P,Q)|\leq F^U(P,Q)$ is known \cite[Property~11]{MR4942599}. We give a derivation to identify its equality case. The trace-norm inequality yields:

\begin{align*}
\left|F_R(P,Q)\right|
&=
\left|
\operatorname{Tr}
\left(
\sqrt{R^{1/2}PR^{1/2}}\,
R^{-1}\,
\sqrt{R^{1/2}QR^{1/2}}
\right)
\right|
\\[4pt]
&\leq
\left\|
\sqrt{R^{1/2}PR^{1/2}}\,
R^{-1}\,
\sqrt{R^{1/2}QR^{1/2}}
\right\|_1
\\[4pt]
&=
\operatorname{Tr}
\sqrt{
\left(
\sqrt{R^{1/2}PR^{1/2}}\,
R^{-1}\,
\sqrt{R^{1/2}QR^{1/2}}
\right)
\left(
\sqrt{R^{1/2}QR^{1/2}}\,
R^{-1}\,
\sqrt{R^{1/2}PR^{1/2}}
\right)
}
\\[4pt]
&=
\operatorname{Tr}
\sqrt{
\sqrt{R^{1/2}PR^{1/2}}\,
R^{-1}
\left(R^{1/2}QR^{1/2}\right)
R^{-1}\,
\sqrt{R^{1/2}PR^{1/2}}
}
\\[4pt]
&=
\operatorname{Tr}
\sqrt{
\sqrt{R^{1/2}PR^{1/2}}\,
R^{-1/2}QR^{-1/2}\,
\sqrt{R^{1/2}PR^{1/2}}
}
\\[4pt]
&=
F^U
\left(
R^{1/2}PR^{1/2},
R^{-1/2}QR^{-1/2}
\right)
\\[4pt]
&=
\left\|\sqrt{P}\sqrt{Q}\right\|_{1}
\\[4pt]
&=
F^U(P,Q),
\end{align*}

where the penultimate equality follows from Lemma~\ref{lem:uhlmann-factorization} with $A=R^{1/2}\sqrt P$ and $B=R^{-1/2}\sqrt Q$. This analysis combined with Lemma~\ref{lem:trace-norm-equality} yields the necessary and sufficient condition for reduction to the Uhlmann fidelity, which we express as follows:

\begin{theorem}\label{thm:equality}
For $P,Q\succeq0$ and $R\succ0$, set
\[
 S=\sqrt{R^{1/2}PR^{1/2}},\quad T=\sqrt{R^{1/2}QR^{1/2}},\quad
 X=R^{-1/2}SR^{-1/2},\quad Y=R^{-1/2}TR^{-1/2}.
\]
Then
\begin{equation}\label{eq:equality}
 |F_R(P,Q)|=F^U(P,Q)
 \quad\Longleftrightarrow\quad F_R(P,Q)=F^U(P,Q)
 \quad\Longleftrightarrow\quad [X,Y]=0.
\end{equation}
Equivalently, equality holds if and only if $SR^{-1}T\succeq0$.
\end{theorem}
\begin{proof}
The preceding calculation and Lemma~\ref{lem:trace-norm-equality} give
$|F_R(P,Q)|=F^U(P,Q)$ if and only if $SR^{-1}T=e^{i\theta}K$
for some $K\succeq0$. Since $SR^{-1}T=R^{1/2}XYR^{1/2}$,
\[
 \Tr(XY)=e^{i\theta}\Tr(R^{-1/2}KR^{-1/2}).
\]
The left side is nonnegative real. If $K\ne0$, the trace on the right is
strictly positive, forcing $e^{i\theta}=1$; if $K=0$, the product already
vanishes. Thus the phase condition is equivalent to $SR^{-1}T\succeq0$.
Congruence by $R^{-1/2}$ makes this equivalent to $XY\succeq0$.
For positive semidefinite $X,Y$, their product is positive semidefinite
exactly when they commute. In this case the trace equals the trace norm,
which is $F^U(P,Q)$ by the calculation above.
\end{proof}

\section*{The Generalized Bures--Wasserstein Distance and the Data Processing Inequality}

Channel monotonicity can fail for alternative fidelities: Mendon\c{c}a et al.\
give a partial-trace counterexample for their fidelity measure
\cite[Sec.~III~D]{Mendonca2008}. Here we establish failure for
Afham--Ferrie's reference-dependent construction with the reference
transported through the same channel.

Throughout this section, $\Phi$ is a quantum channel and the reference states $R$ and $\Phi(R)$ are positive definite, so that both generalized fidelities are defined.

In this section, we disprove the unrestricted data processing inequality considered in the first open problem of \cite{MR4942599}. In fact, we show that $B_R(P,Q) \geq B_{\Phi(R)} (\Phi(P), \Phi(Q))$ does not hold in general in every dimension $n\geq 2$. The counterexamples for dimensions $n \geq 3$ are provided in a constructive way, while the qubit counterexample is given separately. 
 
\begin{theorem}\label{thm:higher-dpi}
For every $n\geq3$, there are positive definite density matrices $P_1,Q_1,R_1$
and a quantum channel $\Phi$ with positive definite outputs such that
$B_{R_1}(P_1,Q_1)<B_{\Phi(R_1)}(\Phi(P_1),\Phi(Q_1))$.
The same strict inequality holds for the distances $b$.
\end{theorem}
\begin{proof}
Start with the positive definite matrices
\begin{equation}\label{eq:higher-dpi-witness}
 A_3=\begin{pmatrix}9&1&0\\1&1&0\\0&0&1\end{pmatrix},\quad
 C_3=\begin{pmatrix}1&-2&0\\-2&5&0\\0&0&1\end{pmatrix},\quad
 B_3=\begin{pmatrix}1&2&0\\2&5&0\\0&0&1\end{pmatrix}.
\end{equation}
Their leading principal minors are respectively $(9,8,8)$, $(1,1,1)$ and
$(1,1,1)$, and $\Tr(A_3C_3B_3)=-5$. For $n>3$, use
$A_n=A_3\oplus I_{n-3}$, $B_n=B_3\oplus I_{n-3}$ and
$C_n=C_3\oplus n^{-1}I_{n-3}$. Then
$\Tr(A_nC_nB_n)=-5+(n-3)/n<0$.
Set
\[
 \widetilde P=C_n^{1/2}A_n^2C_n^{1/2},\qquad
 \widetilde Q=C_n^{1/2}B_n^2C_n^{1/2},\qquad
 \widetilde R=C_n^{-1},
\]
and normalize each to trace one, obtaining $P,Q,R\succ0$. For $a,b,c>0$, homogeneity gives
\[
 F_{cR}(aP,bQ)=\sqrt{ab}F_R(P,Q),
\]
and consequently
$F_R(P,Q)=\alpha<0$. For $n=3$, the three normalization constants are
$53,103,7$, so $\alpha=-5/\sqrt{53\cdot103}$.

For $0<\epsilon<1/2$, define the measurement effects
\begin{equation}\label{eq:higher-dpi-povm}
 \begin{aligned}
 E_1&=\operatorname{diag}(1-2\epsilon,\epsilon,\epsilon)\oplus0_{n-3},\\
 E_2&=\operatorname{diag}(\epsilon,1-2\epsilon,\epsilon)\oplus0_{n-3},\\
 E_3&=\operatorname{diag}(\epsilon,\epsilon,1-2\epsilon)\oplus I_{n-3}.
 \end{aligned}
\end{equation}
They are positive semidefinite and sum to $I_n$, so
$\Phi_\epsilon(Z)=\Tr(E_1Z)P+\Tr(E_2Z)Q+\Tr(E_3Z)R$ is CPTP.
Every density-matrix output is positive definite, being a convex combination
of the positive definite prepared states. Let
\[
 P_1=|1\rangle\langle1|,\qquad Q_1=|2\rangle\langle2|,\qquad
 R_1^\epsilon=\operatorname{diag}(\epsilon,\epsilon,1-2\epsilon)
 \quad(n=3).
\]
For $n>3$, put $2\epsilon/(n-1)$ in every diagonal position of
$R_1^\epsilon$ except the third, which remains $1-2\epsilon$.
Thus $R_1^\epsilon\succ0$ and $\Tr R_1^\epsilon=1$ in every dimension.
As $\epsilon\downarrow0$, the outputs of $P_1,Q_1,R_1^\epsilon$ tend to
$P,Q,R$, respectively. Since $P_1Q_1=0$, the pure-state formula gives
$F_{R_1^\epsilon}(P_1,Q_1)=0$, whereas continuity at the positive definite
output reference gives
\[
 F_{\Phi_\epsilon(R_1^\epsilon)}(\Phi_\epsilon(P_1),\Phi_\epsilon(Q_1))
 \longrightarrow\alpha<0.
\]
Choose a sufficiently small $\epsilon>0$ giving a strict DPI violation,
and fix this parameter and channel. Replace $P_1,Q_1$ by
$P_{1,\delta}=(1-\delta)P_1+\delta I_n/n$ and
$Q_{1,\delta}=(1-\delta)Q_1+\delta I_n/n$.
For all sufficiently small $\delta>0$, continuity preserves the strict
violation while both inputs become positive definite. The reference
$R_1^\epsilon$ stays fixed. Taking square roots preserves the strict gap.
\end{proof}

\begin{theorem}\label{thm:qubit-dpi-counterexample}
The inequality
$B_R(P,Q)\geq B_{\Phi(R)}(\Phi(P),\Phi(Q))$
does not hold in general in dimension two, even for a unital channel
and positive definite input and output references. More precisely, let
\[
 P=\frac1{10}\begin{pmatrix}1&3\\3&9\end{pmatrix},\qquad
 Q=\frac15\begin{pmatrix}4&2\\2&1\end{pmatrix},\qquad
 R_\varepsilon=\begin{pmatrix}1-\varepsilon&0\\0&\varepsilon\end{pmatrix},
\]
and define
\[
 \Phi(T)=\frac{31}{32}T+\frac1{32}XTX,\qquad
 X=\begin{pmatrix}0&1\\1&0\end{pmatrix}.
\]
For every $0<\varepsilon\leq1/1000$, these are qubit density matrices,
$R_\varepsilon,\Phi(R_\varepsilon)\succ0$, and
\[
 B_{R_\varepsilon}(P,Q)<\frac35
 <B_{\Phi(R_\varepsilon)}(\Phi(P),\Phi(Q)).
\]
Thus the generalized Bures--Wasserstein distance also fails the
data processing inequality. For the concrete value $\varepsilon=1/1000$,
\[
\begin{aligned}
 B_{R_\varepsilon}(P,Q)&\approx0.5901764973,\\
 B_{\Phi(R_\varepsilon)}(\Phi(P),\Phi(Q))&\approx0.6029697904,
\end{aligned}
\]
so the two sides are separated by a strict numerical gap. The exact
algebraic expressions and an independently reproducible evaluation are
reported in the proof and in the Supplementary Material (the scripts and recorded output).
\end{theorem}
\begin{proof}
The states $P,Q$ have rank one and trace one, and $\Phi$ is a convex
combination of unitary conjugations. Moreover,
$\Phi(R_\varepsilon)=\operatorname{diag}(31-30\varepsilon,1+30\varepsilon)/32
\succ0$. The qubit square-root formula gives
\[
 F_{R_\varepsilon}(P,Q)
 =\frac{2+\varepsilon}{\sqrt{2(4+29\varepsilon-24\varepsilon^2)}}
 >\frac7{10},
\]
where squaring reduces the strict inequality to
$4-1221\varepsilon+1226\varepsilon^2>0$.
For the output fidelity, the same formula gives
\[
 F_{\Phi(R_\varepsilon)}(\Phi(P),\Phi(Q))
 =\frac{A+B-707/1600}{2\sqrt{AB}},
\]
where
\[
 S=\sqrt{961+27900\varepsilon-27900\varepsilon^2},\qquad
 A=\frac{95+450\varepsilon+S}{640},\qquad
 B=\frac{1955-1350\varepsilon+3S}{2560}.
\]
The bounds $31\leq S<32$ and $S\leq31+450\varepsilon$ imply
$(A,B)\in[63/320,1/4]\times[3/4,4/5]$.
On this rectangle,
$H(u,v)=(u+v-707/1600)^2-49uv/25$
decreases in $u$ and increases in $v$, so
\[
 H(A,B)\leq H(63/320,4/5)=-\frac{27}{40000}<0.
\]
Since $A+B-707/1600>0$, the output fidelity is less than $7/10$.
At $\varepsilon=1/1000$, the input and output fidelities are respectively
$667\sqrt{1142}/31976\approx0.7049117514$ and
$0.6985151048$. Both fidelities are real, and $B_R=2-2\operatorname{Re}F_R$
proves the claim; taking square roots gives the distance inequality.
\end{proof}

The strict gap also persists for positive definite input states. Indeed,
for any fixed $0<\varepsilon\leq1/1000$, replace $P,Q$ by
$(1-\delta)P+\delta I_2/2$ and $(1-\delta)Q+\delta I_2/2$.
Continuity at the positive definite references implies a violation
for all sufficiently small $\delta>0$.

The preceding theorems show that the data processing inequality fails in every dimension $n\geq2$. Nevertheless, Corollary~\ref{cor:fidelity} yields sufficient conditions for the qubit case. 

\begin{corollary}\label{cor:qubit-dpi-sign}
  For qubit states $P,Q,R$ with $R\succ0$, if the real part of $F_R(P,Q)$ is negative or zero, then for any channel $\Phi$ from qubit states to qubit states with $\Phi(R)\succ0$, we have
\[
 B_R(P,Q) \geq B_{\Phi(R)} (\Phi(P), \Phi(Q)).
\]

  \begin{proof}
By monotonicity of the root Uhlmann fidelity under channels
\cite[Theorem~3.27]{Watrous}, we have:

\[
\begin{aligned}
\operatorname{Re}\!\left(
F_{\Phi(R)}(\Phi(P),\Phi(Q))
\right)
&=
\frac{
\bigl(F^U(\Phi(P),\Phi(Q))\bigr)^2
+\bigl(F^U(\Phi(R),\Phi(P))\bigr)^2
+\bigl(F^U(\Phi(R),\Phi(Q))\bigr)^2
-1
}{
2F^U(\Phi(R),\Phi(P))F^U(\Phi(R),\Phi(Q))
}
\\[4pt]
&\geq
\frac{
\bigl(F^U(P,Q)\bigr)^2
+\bigl(F^U(R,P)\bigr)^2
+\bigl(F^U(R,Q)\bigr)^2
-1
}{
2F^U(\Phi(R),\Phi(P))F^U(\Phi(R),\Phi(Q))
}
\\[4pt]
&\geq
\frac{
\bigl(F^U(P,Q)\bigr)^2
+\bigl(F^U(R,P)\bigr)^2
+\bigl(F^U(R,Q)\bigr)^2
-1
}{
2F^U(R,P)F^U(R,Q)
}
\\[4pt]
&=
\operatorname{Re}\!\left(F_R(P,Q)\right).
\end{aligned}
\]

Note that the last inequality holds, since we have assumed that the real part of $F_R(P,Q)$ is negative or zero.
  \end{proof}
\end{corollary}

\begin{corollary}\label{cor:qubit-dpi-min}
    For qubit states $P,Q,R$ with $R\succ0$, if $F^U(P,Q) \leq \min \{ F^U(R,Q), F^U(R,P)\}$, then for any channel $\Phi$ from qubit states to qubit states with $\Phi(R)\succ0$, we have
\[
 B_R(P,Q) \geq B_{\Phi(R)} (\Phi(P), \Phi(Q)).
\]

    \begin{proof}
        For simplicity, set $c=F^U(P,Q)$, $b=F^U(R,Q)$, and $a=F^U(R,P)$. By Uhlmann data processing, there are $x,y,z\geq0$ such that
\[
\begin{aligned}
 F^U(\Phi(P),\Phi(Q))&=c+z,\\
 F^U(\Phi(R),\Phi(P))&=a+x,\\
 F^U(\Phi(R),\Phi(Q))&=b+y.
\end{aligned}
\]
        
            Moreover, 

        $$\operatorname{Re}\!\left(
F_{\Phi(R)}(\Phi(P),\Phi(Q))
\right)=\frac{(c+z)^2+(a+x)^2+(b+y)^2-1}{2(a+x)(b+y)} \geq \frac{c^2+(a+x)^2+(b+y)^2-1}{2(a+x)(b+y)}.$$

Viewing $\frac{c^2+(a+x)^2+(b+y)^2-1}{2(a+x)(b+y)}$ as a function in terms of $x,y$, say $f(x,y)$, we will have:

$$f_x= \frac{(a+x)^2-(b+y)^2-c^2+1}{2(a+x)^2(b+y)}, $$

$$f_y=\frac{(b+y)^2-(a+x)^2-c^2+1}{2(a+x)(b+y)^2}.$$

Since $c \leq \min \{a,b\}$, $x, y \geq 0$, and $(a+x)^2, (b+y)^2 \leq 1$, it follows that $f_x,f_y\geq0$. 

Therefore, 

$$\operatorname{Re}\!\left(
F_{\Phi(R)}(\Phi(P),\Phi(Q))
\right) \geq \frac{c^2+(a+x)^2+(b+y)^2-1}{2(a+x)(b+y)} \geq \frac{c^2+a^2+b^2-1}{2ab}= \operatorname{Re}\!\left(F_R(P,Q)\right).$$

    \end{proof}
\end{corollary}

\begin{corollary}\label{cor:qubit-dpi-band}
Let $P,Q,R$ be qubit density matrices with $R\succ0$. Write
$a=F^U(R,P)$, $b=F^U(R,Q)$ and $c=F^U(P,Q)$. If
\[
 |a^2-b^2|\leq1-c^2,
\]
then, for every qubit-to-qubit quantum channel $\Phi$ with
$\Phi(R)\succ0$,
\[
 B_R(P,Q)\geq B_{\Phi(R)}(\Phi(P),\Phi(Q)).
\]
The same conclusion holds for the generalized Bures--Wasserstein
distance $\sqrt{B_R}$.
\end{corollary}
\begin{proof}
As in the preceding proof, Uhlmann data processing gives output
fidelities $a+x,b+y,c+z$, with $x,y,z\geq0$. Put $d=1-c^2$ and
$g(a,b,c)=(a^2+b^2+c^2-1)/(2ab)$. Since $a,b>0$,
\[
 \begin{aligned}
 &2ab(a+x)(b+y)\bigl[g(a+x,b+y,c)-g(a,b,c)\bigr]\\
 &\quad=bx(a^2-b^2+d)+ay(b^2-a^2+d)\\
 &\qquad\quad+ab(x-y)^2+\bigl[d-(a-b)^2\bigr]xy\geq0.
 \end{aligned}
\]
Indeed, $d\geq|a^2-b^2|\geq(a-b)^2$ makes every term nonnegative.
Increasing $c$ also increases $g$, so the qubit fidelity formula yields
\[
 \operatorname{Re}F_{\Phi(R)}(\Phi(P),\Phi(Q))
 =g(a+x,b+y,c+z)\geq g(a,b,c)=\operatorname{Re}F_R(P,Q).
\]
Now use $B_R=2-2\operatorname{Re}F_R$ and take square roots.
\end{proof}

Corollary~\ref{cor:qubit-dpi-band} subsumes
Corollaries~\ref{cor:qubit-dpi-sign} and \ref{cor:qubit-dpi-min}.
The first assumes $a^2+b^2\leq1-c^2$, which implies the band condition.
For the second, $c\leq\min(a,b)$ gives
$|a^2-b^2|\leq1-\min(a^2,b^2)\leq1-c^2$.

\section*{Counterexamples to Convexity}

The root Uhlmann fidelity is jointly concave
\cite[Corollary~3.26]{Watrous}, so the ordinary squared Bures--Wasserstein
distance $d_{\mathrm{BW}}^2$ is jointly convex under affine mixtures of
states. This provides the benchmark for the convexity question in
\cite[Sec.~X, Open Problem~2]{MR4942599}. A positive special case of the
generalized construction is $R=cI$, $c>0$: direct substitution gives
$F_{cI}(P,Q)=\Tr(\sqrt P\sqrt Q)$, the Holevo affinity, whose joint
concavity follows from Lieb's trace concavity theorem \cite{Lieb1973}.
Thus $B_{cI}$ is jointly convex. We show that for other fixed references
neither $B_R$ nor $b_R$ is convex in $P$ in general, in every dimension
$n\geq2$. The constructions for $n\geq3$ and for qubits are given
separately. Convexity here always means convexity under affine mixtures;
the distance and its square require separate arguments. 

\begin{theorem}\label{thm:higher-convexity}
For every $n\geq3$, there are positive definite density matrices
$P_1,P_2,Q,R$ such that, with $M=(P_1+P_2)/2$,
\[
 B_R(M,Q)>\frac{B_R(P_1,Q)+B_R(P_2,Q)}2,\qquad
 b_R(M,Q)>\frac{b_R(P_1,Q)+b_R(P_2,Q)}2.
\]
Thus neither the distance $b_R$ nor its square $B_R$ is convex in $P$ with $Q,R$ fixed.
\end{theorem}
\begin{proof}
We construct states whose endpoint generalized fidelities are positive
and whose midpoint generalized fidelity is negative. This violates concavity
of the real part of generalized fidelity. Use the matrices
\begin{equation}\label{eq:higher-convexity-witness}
 \begin{aligned}
 A_1&=\begin{pmatrix}1&1/2&-2\\1/2&1&-2\\-2&-2&6\end{pmatrix}\oplus0_{n-3},\\
 A_2&=\begin{pmatrix}1&1/2&2\\1/2&1&2\\2&2&6\end{pmatrix}\oplus0_{n-3},\\
 C&=\operatorname{diag}(1,\epsilon,\ldots,\epsilon),\qquad
 B=vv^\top,\quad v=(1,-3/2,0,\ldots,0)^\top.
 \end{aligned}
\end{equation}
The leading $3\times3$ blocks of $A_i$ are positive definite, with leading
principal minors $1,3/4,1/2$, and $B\succeq0$. Direct squaring gives
\[
 D:=\sqrt{\frac{A_1^2+A_2^2}{2}}
 =\begin{pmatrix}
 (\sqrt{41}+1)/4&(\sqrt{41}-1)/4&0\\
 (\sqrt{41}-1)/4&(\sqrt{41}+1)/4&0\\
 0&0&2\sqrt{11}
 \end{pmatrix}\oplus0_{n-3}.
\]
Consequently,
\begin{equation}\label{eq:higher-convexity-signs}
 \begin{aligned}
 \Tr(A_iCB)&=\frac14+\frac{3\epsilon}{2}>0,\qquad i=1,2,\\
 \Tr(DCB)&=\frac{5-\sqrt{41}}8+
              \frac{3\epsilon(5+\sqrt{41})}{16}<0
 \end{aligned}
\end{equation}
whenever $0<\epsilon<2(\sqrt{41}-5)/(3(\sqrt{41}+5))$;
for example, take $\epsilon=1/100$.
The positive semidefinite matrices
$\widetilde P_i=C^{1/2}A_i^2C^{1/2}$ and
$\widetilde Q=C^{1/2}B^2C^{1/2}$ have traces
\[
 t:=\Tr\widetilde P_1=\Tr\widetilde P_2=\frac{21+197\epsilon}{4},\qquad
 u:=\Tr\widetilde Q=\frac{13(4+9\epsilon)}{16}.
\]
Set $P_i=\widetilde P_i/t$, $Q=\widetilde Q/u$, and
$R=C^{-1}/\Tr(C^{-1})$. Homogeneity in $R$ and common normalization of
the endpoints give
\[
 F_R(P_i,Q)=\frac{\Tr(A_iCB)}{\sqrt{tu}}>0,\qquad
 F_R(M,Q)=\frac{\Tr(DCB)}{\sqrt{tu}}<0.
\]
These fidelities are real since the matrices are real. Thus $B_R(M,Q)>2$
while both endpoint squared distances are below $2$. Also $b_R(M,Q)>\sqrt2$
while both endpoint distances are below $\sqrt2$, proving both strict
Jensen violations separately.

Finally fix $\epsilon=1/100$ and replace each $X\in\{P_1,P_2,Q\}$ by
$X_\delta=(1-\delta)X+\delta I_n/n$. The reference stays fixed, and
\[
 \frac{P_{1,\delta}+P_{2,\delta}}2
 =(1-\delta)M+\delta I_n/n.
\]
All three states become positive definite, the midpoint relation is preserved,
and continuity preserves both strict violations for sufficiently small
$\delta>0$.
\end{proof}

\begin{theorem}\label{thm:qubit-convexity-counterexample}
In dimension two, there are positive definite density matrices
$P_1,P_2,Q,R$ such that, with $M=(P_1+P_2)/2$,
\[
 B_R(M,Q)-\frac{B_R(P_1,Q)+B_R(P_2,Q)}2>\frac7{100},
\]
and
\[
 \sqrt{B_R(M,Q)}
 -\frac{\sqrt{B_R(P_1,Q)}+\sqrt{B_R(P_2,Q)}}2>\frac{49}{200}.
\]
Consequently, neither the generalized Bures--Wasserstein distance nor
its square is convex in $P$, even with $Q,R$ fixed. In particular,
neither is jointly convex.
\end{theorem}
\begin{proof}
Take
\[
 R=\frac1{100}\begin{pmatrix}1&0\\0&99\end{pmatrix},\qquad
 P_1=\frac12I_2,\qquad
 P_2=Q=\frac1{100}\begin{pmatrix}90&29\\29&10\end{pmatrix},
\]
so $M=\frac1{200}\begin{pmatrix}140&29\\29&60\end{pmatrix}$.
These matrices have trace one and are positive definite; in particular,
$\det Q=59/10000$ and $\det M=7559/40000$.
Put $x=\sqrt{99}$, $y=\sqrt{59}$ and $z=\sqrt{7559}$.
The qubit square-root formula gives the real fidelities
\[
 \begin{aligned}
 \alpha:=F_R(P_1,Q)
 &=\frac{540+50x+xy+50y}
 {10000\sqrt{(50+x)(540+xy)/500000}}>\frac{11}{20},\\
 \gamma:=F_R(M,Q)
 &=\frac{1561+xz+2xy+yz}
 {20000\sqrt{(3040+xz)(540+xy)/50000000}}<\frac{37}{50},
 \end{aligned}
\]
and $F_R(P_2,Q)=1$. Squaring positive quantities and using
$x^2=99$, $y^2=59$, $z^2=7559$ reduces the two comparisons to
\[
 \begin{aligned}
 -1882118+54460x+115821y+6110xy&>0,\\
 2210696425-5027730xy-620230xz-1927688yz&>0.
 \end{aligned}
\]
Squaring gives the rational bounds
\[
 \frac{9949}{1000}<x<\frac{9950}{1000},\quad
 \frac{7681}{1000}<y<\frac{7682}{1000},\quad
 \frac{86942}{1000}<z<\frac{86943}{1000}.
\]
Substitute the lower bounds in the first expression and the upper bounds
in the second. The resulting lower bounds are respectively
$1624126459/100000>0$ and $587918525203/250000>0$. Hence
\[
 B_R(M,Q)>\frac{13}{25},\qquad
 \frac{B_R(P_1,Q)+B_R(P_2,Q)}2<\frac9{20},
\]
which gives the first gap. For the distance,
\[
 \sqrt{B_R(M,Q)}>\frac{18}{25},\qquad
 \frac{\sqrt{B_R(P_1,Q)}+\sqrt{B_R(P_2,Q)}}2<\frac{19}{40},
\]
because $B_R(P_2,Q)=0$, $(18/25)^2<13/25$ and
$(19/40)^2>9/40$. Their difference exceeds $49/200$. Numerically, the two sides of the
squared-distance convexity inequality are
\[
 B_R(M,Q)\approx0.5202873984,
 \qquad \frac{B_R(P_1,Q)+B_R(P_2,Q)}2\approx0.4489950389,
\]
while the two sides for the distance itself are
\[
 \sqrt{B_R(M,Q)}\approx0.7213095025,
 \qquad \frac{\sqrt{B_R(P_1,Q)}+\sqrt{B_R(P_2,Q)}}2\approx0.4738116920.
\]
The exact radicals above and an independent high-precision check are included
in the Supplementary Material (the scripts and recorded output).
\end{proof}

For the unsquared distance, there is a simpler obstruction valid for
every fixed reference $R\succ0$ in dimension $n\geq2$. Choose orthonormal
eigenvectors $e_1,e_2$ of $R$, and put
$P_1=|e_1\rangle\langle e_1|$, $P_2=Q=|e_2\rangle\langle e_2|$ and
$M=(P_1+P_2)/2$. These states commute with $R$, so the defining fidelity
reduces to classical affinity and
\[
 b_R(M,Q)=\sqrt{2-\sqrt2}>\frac{\sqrt2}{2}
 =\frac{b_R(P_1,Q)+b_R(P_2,Q)}2.
\]
A common sufficiently small depolarization $X\mapsto(1-\delta)X+\delta I_n/n$
of the three states preserves this strict gap and makes them positive
definite. Hence no reference makes $b_R(\cdot,Q)$ convex for every $Q$.
The failure of convexity for $B_R$ in the preceding theorems is the
stronger distinction from ordinary squared Bures--Wasserstein geometry.

\section*{Conclusion}
We establish failure of unrestricted transported-reference data processing
and of convexity in one state, for both $B_R$ and $b_R$, in every dimension
$n\geq2$, including positive definite witnesses. The qubit DPI example uses
a unital channel. For qubits, Corollary~\ref{cor:qubit-dpi-band} gives a
sufficient DPI criterion containing the two preceding conditions.
The nonconvexity of $B_R$ contrasts with joint convexity of the ordinary
squared Bures--Wasserstein distance, while scalar references retain joint
convexity through Holevo affinity. For $b_R$, every reference admits a
convexity violation for a suitable fixed $Q$. Determining all triples and
channels that satisfy transported-reference data processing, and
characterizing the fixed pairs $(Q,R)$ for which $B_R(\cdot,Q)$ is convex
in arbitrary dimensions, remain open.

\section*{Supplementary Material}
The file \path{verify_qubit_examples.py} checks the concrete examples.
An independent verifier, \path{verify_full_paper.py}, checks symbolic
identities and evaluates the defining matrix square roots at 80 decimal digits. Both scripts also check higher-dimensional witnesses.
The accompanying output files record successful runs and dependency versions;
\path{README.md} gives the reproduction commands. Python 3, SymPy and mpmath
are required. These computations supplement the analytic proofs.

\section*{Acknowledgments}
Reza Rajaei thanks Prof. J.A. Ch\'avez-Dom\'inguez for helpful discussions and useful suggestions.

\section*{Author Declarations}
\subsection*{AI Use}
In the course of this work, Wenlin Zhang and Haobo Ma used AI tools
(ChatGPT Pro, GPT-6 Pro; Codex CLI 0.161.0, model gpt-6.1-sol;
Claude Opus 5.5 via Claude Code) to assist with computations, checking
and drafting; they verified all results and take full responsibility
for their contributions.

\subsection*{Conflict of Interest}
The authors have no conflicts to disclose.

\subsection*{Author Contributions}
\textbf{Haobo Ma:} Formal analysis; Investigation; Software; Validation;
Construction with Wenlin Zhang of the qubit counterexamples
(Theorems~\ref{thm:qubit-dpi-counterexample} and
\ref{thm:qubit-convexity-counterexample}), separately from the
higher-dimensional constructions; Writing -- original draft (these
counterexamples and Corollary~\ref{cor:qubit-dpi-band});
Writing -- review \& editing.

\textbf{Reza Rajaei:} Conceptualization; Formal analysis; Investigation;
Construction of the higher-dimensional counterexamples
(Theorems~\ref{thm:higher-dpi} and \ref{thm:higher-convexity});
Writing -- original draft (the earlier preprint \cite{Rajaei2026v1} and
the higher-dimensional convexity construction); Writing -- review \& editing;
Project administration.

\textbf{Wenlin Zhang:} Formal analysis; Investigation; Software; Validation;
Construction with Haobo Ma of the qubit counterexamples
(Theorems~\ref{thm:qubit-dpi-counterexample} and
\ref{thm:qubit-convexity-counterexample}), separately from the
higher-dimensional constructions; Writing -- original draft (these
counterexamples and Corollary~\ref{cor:qubit-dpi-band});
Writing -- review \& editing.

\section*{Data Availability}
The exact matrices and inequalities supporting the findings are given in
the article. The verification code (\path{verify_qubit_examples.py} and
\path{verify_full_paper.py}), recorded outputs and reproduction instructions
are available as the Supplementary Material accompanying this article.

\bibliography{references}
\bibliographystyle{amsalpha}

\end{document}